%% file: main.tex
\documentclass[letterpaper, 10 pt, conference]{ieeeconf}  

\usepackage{graphicx}
\usepackage{amsmath}
\usepackage{amssymb}
\usepackage{algorithm,algorithmic}
\usepackage{euscript}
\usepackage{verbatim}
\usepackage{amssymb}
\usepackage{cite}
\usepackage{color}
\usepackage{cancel}
\usepackage{calligra}
\usepackage{pgf}
\usepackage{tikz}
\usepackage{ctable}
\usetikzlibrary{arrows,automata}
\usetikzlibrary{shapes}
\usetikzlibrary{positioning}
\usepackage{multirow}
\usepackage{url}

{\normalfont}{\normalfont}
{\normalfont}{\normalfont}
\newtheorem{remark}{Remark}{\normalfont}{\normalfont}
\newtheorem{theorem}{Theorem}
\newtheorem{assumption}{Assumption}
\newtheorem{proposition}{Proposition}
\newtheorem{lemma}{Lemma}

\let\labelindent\relax
\usepackage{enumitem}
\renewcommand{\theenumi}{\arabic{enumi}}

\renewcommand{\theenumii}{\arabic{enumii}}

\renewcommand{\theenumiii}{\arabic{enumiii}}

\renewcommand{\theenumiv}{\arabic{enumiv}}

\title{\LARGE \bf Model predictive control with dynamic move blocking}

\author{Valentina Breschi, Simone Formentin and Alberto Leva
	\thanks{This project was partially supported by the Italian Ministry of University and Research under the PRIN'17 project \textquotedblleft Data-driven learning of constrained control systems \textquotedblright, contract no. 2017J89ARP.}
\thanks{V. Breschi is with Department of Electrical Engineering, Eindhoven University of Technology, Eindhoven, Netherlands. S. Formentin and A. Leva are with Dipartimento di Elettronica, Informatica e Bioingegneria (DEIB), Politecnico di Milano, Via G. Ponzio 34/5, 20133 Milano, Italy. Corresponding author: {\tt\small simone.formentin@polimi.it}}
}

\begin{document}
	
\maketitle
\thispagestyle{empty}
\pagestyle{empty}

\begin{abstract}
Model Predictive Control (MPC) has proven to be a powerful tool for the control of systems with constraints. Nonetheless, in many applications, a major challenge arises, that is finding the optimal solution within a single sampling instant to apply a receding-horizon policy. In such cases, many suboptimal solutions have been proposed, among which the possibility of ``blocking'' some moves a-priori. In this paper, we propose a \textit{dynamic} approach to move blocking, to exploit the solution already available at the previous iteration, and we show not only that such an approach preserves asymptotic stability, but also that the decrease of performance with respect to the ideal solution can be theoretically bounded.
\end{abstract}

\input{./sections/01-Introduction.tex}

\input{./sections/02-Problem-statement.tex}
\input{./sections/03-DMB-MPC.tex}
\input{./sections/04-App-example.tex}

\input{./sections/05-Conclusions.tex}

\bibliography{2023-CDC-DMB-MPC}

\end{document}

%% file: sections/01-Introduction.tex
\section{Introduction}
\label{sec:Introduction}

Model Predictive Control (MPC) is fundamental to handle constrained control problems providing formal guarantees~\cite{allgower2012nonlinear,borrelli2017predictive}. MPC most frequently relies on the Receding Horizon (RH) approach: an Optimisation Problem (OP) is solved at each step to yield a vector of future controls, only the first one is applied, and the whole procedure is repeated for the next step.

Sometimes, however, solving the OP within one control timestep can be computationally infeasible, and simplifying for brevity, one can choose among the following alternatives:

\begin{itemize}

\item increase the timestep --- but this can be prohibited by the physics of the control problem;

\item resort to \emph{explicit} MPC~\cite{alessio2009survey} --- but this can entail too demanding memory requirements and would not be feasible for large scale problems;

\item simplify the OP, for example via local linearisations --- but this makes global properties hard to enforce \cite{berberich2022linear};

\item reduce the OP dimensionality by shrinking the control horizon or by having some future controls depend on others, which is called ``Move Blocking'' (MB) \cite{shekhar2015optimal,cagienard2007move}.

\end{itemize}

Notice that to perform less optimisations one could also apply not just the first computed control sample as in the RH case, but up to all those that cover the (future) control horizon $N$. This is called OL-MPC for Open-Loop MPC, as the loop actually closes just at each $N$-th step, but does not remove the constraint of solving the OP within one timestep.

In our research, we propose a novel combination of OL-MPC and MB, accepting that the OP cannot be solved in one timestep and just requiring an overbound $N_{OP}$ of the number of steps required. In detail

\begin{itemize}

\item after the $k$-th optimisation we apply $N-N_B$ control samples, where $N_B \geq N_{OP}$ is the \emph{blocking horizon},

\item we start the $(k+1)$-th optimisation $N_{OP}$ steps before the above samples are exhausted, hence overlapping the control horizons,

\item and we constrain the first $N_{OP}$ controls from the $(k+1)$-th optimisation to equal those coming from the $k$-th one, applied -- as said -- while the $(k+1)$-th is running.

\end{itemize}

The main novelty is that in doing so we make the MB mechanism \textit{stateful}, which is why the name our technique DMB for \textit{Dynamic MB}.

In this paper, analogously to what has been done for other (static) MB strategies (see, e.g., \cite{gondhalekar2007recursive,schwickart2016flexible}), we analyse the DMB strategy so as to assess the level of sub-optimality with respect to a purely RH realisation of the same MPC scheme, that we take as reference but is infeasible owing to computational limits. We further show that asymptotic stability is preserved under some mild assumptions, and illustrate the limited decrease of performance on a simple numerical example. A preliminary version of this idea was proposed in \cite{leva2022overlapping}, where however no theoretical analysis of the proposed scheme was provided.

The remainder of the paper is as follows. In Section \ref{sec:Problem-statement}, we formally state the problem of interest and introduce the main terminology and notation. Section \ref{sec:DMB-MPC} discusses the DMB-MPC approach and illustrates its theoretical properties. A simple numerical example is proposed in Section \ref{sec:POC-example}. The paper is ended by some concluding remarks.

%% file: sections/02-Problem-statement.tex
\section{Problem statement}
\label{sec:Problem-statement}


Consider the \emph{nonlinear} discrete-time dynamical system
\begin{equation}\label{eq:S}
	x(t+1)=f(x(t),u(t)),~~x(0)=x_{0},
\end{equation}
where $x(t) \in \mathcal{X} \subseteq \mathbb{R}^n$ is the state of the system at time $t \in \mathbb{R}^{n}$, laying in a \emph{finite dimensional} space $X$, while $u(t) \in U \subseteq \mathbb{R}^m$ is a controllable input and $f: X \times U \rightarrow X$. For a given state $x$, our goal would be to find a feedback control action $u=\mu_{\infty}(x)$ that minimizes the infinite horizon cost
\begin{equation}\label{eq:infinite_cost}
	J_{\infty}(x,u)=\sum_{k=0}^{\infty} \ell(x_{u}(k),u(k)),~~x_{u}(0)=x,
\end{equation}
where $x_{u}(k)$ denotes the state realized for some control sequence\footnote{$\mathcal{U}$ indicates the space of feasible control sequences.} $u \in \mathcal{U}$, and $\ell: X \times U \rightarrow \mathbb{R}_{0}^{+}$ is the running cost of the problem. The latter is here assumed to be \emph{positive definite} and \emph{proper} with respect to some compact set $A$.

However, minimizing \eqref{eq:infinite_cost} is often infeasible in practice, motivating the introduction of \emph{model predictive control} (MPC) strategies that focus on solving the following \textit{finite-horizon} alternative:
\begin{subequations}\label{eq:MPC}
	\begin{align}
		&\underset{u}{\mbox{minimize}}~~J_{N}(x,u)\\
		& \qquad ~ \mbox{s.t. }~x_{u}(k\!+\!1)\!=\!f(x_{u}(k),u(k)),~~ \forall k \in \mathcal{I}_{N},\\
		&\qquad \qquad ~ x_u(0)\!=\!x,\\
		&\qquad \qquad ~ x_{u}(k) \in X,~u(k) \in U,~~\forall k \in \mathcal{I}_{N},
	\end{align}
where $\mathcal{I}_{N=\{0,\ldots,N-1\}}$ and the cost is given by
\begin{equation}\label{eq:RH_cost}
	J_{N}(x,u)=\sum_{k=0}^{N-1} \ell(x_{u}(k),u(k)).
\end{equation}
\end{subequations}
Once the problem is solved, only the first control action is fed to the system, whereas the rest of the optimized input sequence is discarded. The problem in \eqref{eq:MPC} needs then to be solved again, with updated initial state, at the following sampling instant. By introducing the optimal value function $V_{N}(x)$ associated to this problem, namely,
\begin{equation}\label{eq:RH_ovf}
	V_{N}(x)=\inf_{u \in \mathcal{U}} J_{N}(x,u),
\end{equation}
a direct consequence of the Bellman's optimality principle is that the control law is finally selected as
\begin{equation}\label{eq:RH_law}
	\mu_{N}(x)=\arg\min_{u \in U}~\{V_{N-1}(f(x,u))+\ell(x,u)\}.
\end{equation}

Nonetheless, especially when $N$ is large, also this simplified strategy can eventually become unfeasible in practice, since it requires the solution of an optimization problem in real-time at each time step. This is indeed the case of fast sampled systems, where a single sampling instant might not be sufficient to solve the entire optimization problem. Under the assumption that a number of steps $N_{OP}(t)>1$ \footnote{We stress here that $N_{OP}(t)$ might potentially be time-varying.} are needed to solve \eqref{eq:MPC} at each time instant $t$, our objective in this paper is to propose a move-blocking strategy that allows us to $(i)$ approximate the solution of \eqref{eq:MPC} \textit{with a certified level of sub-optimality}, while $(ii)$ being compatible with the computational limits of the application at hand.

%% file: sections/03-DMB-MPC.tex
\section{Dynamic move blocking}
\label{sec:DMB-MPC}

In this work, we focus on two elements to reduce the computational burden characterizing standard receding horizon strategies: $(i)$ the degrees of freedom of \eqref{eq:MPC} and $(ii)$ the input actually fed to the system.

To limit the degrees of freedom and, thus, the computational complexity of the optimal control problem at time $t$, we leverage the idea of \textit{move-blocking} strategies \cite{} and consider the following \textquotedblleft variation\textquotedblright \ of the standard MPC problem:
\begin{subequations}\label{eq:blocked_MPC}
	\begin{align}
		&\underset{u}{\mbox{minimize}}~~J_{N}(x(t),u) \label{eq:blocked_MPCcost}\\
		& ~~\mbox{s.t. }~x_{u}(k\!+\!1)\!=\!f(x_{u}(n),u(n)),~~ \forall k \in \mathcal{I}_{N},\\
		&\qquad~~ x_u(0)\!=\!x(t),\\
		& \qquad~~ u(k)\!=\!\tilde{u}(N\!\!-\!\!N_{B}\!+\!k),~k\!=\!0,\ldots,N_{B}\!-\!1,\\
		&\qquad~~ x_{u}(k) \in X,~u(n) \in U,~~\forall k \in \mathcal{I}_{N},
	\end{align}
	where $x(t)$ is the measured (or estimated) state of the system at the time when the optimization is carried out, and $N_{B} \geq N_{OP}$ dictates the number of moves that are \emph{overlapped}, since they are here blocked to the optimal control actions obtained by solving this problem at the previous optimization step (here denoted as $\tilde{u} \in \mathcal{U}$). As the values of the ``blocked'' inputs depends on the optimization at the previous time instant, we will call such a strategy MPC with \textit{dynamic move blocking} (MPC-DMB) hereafter.
	\begin{remark}[$N_{B}$ and $N$]
		Clearly, the number $N_{B}$ of blocked actions must also satisfy $N_{B} \leq N$. The special case of $N_{B}=N$ would imply that the past sequence is fully used to control the system and only the action $u(t+N-1)$ is optimally computed. As we will show next, this choice is far from being optimal, since the last control action is computed without \textquotedblleft looking into the predicted future\textquotedblright \ of the state, ultimately jeopardizing closed-loop stability.
	\end{remark}
\end{subequations}

\subsection{Properties of MPC-DMB}
We now discuss some of the properties of the control law originated by the MPC-DMB approach, namely
\begin{subequations}
\begin{equation}\label{eq:blocked_law}
	\mu_{N}^{\mathrm{DMB}}(x)=\begin{bmatrix}
		\tilde{u}(N-N_{B})\\
		\vdots\\
		\tilde{u}(N-1)\\
		\mu_{N_{B}}(x_{N_{B}})
	\end{bmatrix},
\end{equation}
where $x_{N_{B}}$ compactly denotes the $N_{B}$-step ahead prediction of the system's state starting from the initial condition $x$ and using the blocked inputs $\tilde{u}(N-N_{B}+n)$, for $n=1,\ldots,N_{B}-1$, and
\begin{equation}\label{eq:DMB_lawLast}
	\mu_{N_{B}}(x_{N_{B}})\!=\!\arg\min_{u \in U}\{V_{N\!-\!N_{B}\!-\!1}(f(x_{N_{B}},u))\!+\!\ell(x_{N_{B}},u)\}.
\end{equation}
\end{subequations}

In the next Lemma, we will show that blocking the first $N_{B}$ actions in \eqref{eq:blocked_MPC} is equivalent to reformulate the problem as a standard MPC with \textit{reduced} horizon
\begin{subequations}\label{eq:blocked_MPCsimplified}
	\begin{align}
	&\underset{u}{\mbox{minimize}}~~J_{N-N_{B}}(x(t+N_{B}),u) \\
	& ~~\mbox{s.t. }~x_{u}(k\!+\!1)\!=\!f(x_{u}(k),u(k)),~~ \forall k \in \mathcal{I}_{N-N_{B}},\\
	&\qquad~~ x_u(0)\!=\!x(t+N_{B}),\\
	&\qquad~~ x_{u}(k) \in X,~u(k) \in U,~~\forall k \in \mathcal{I}_{N-N_{B}},
	\end{align}
	where $\mathcal{I}_{N-N_{B}}=\{0,\ldots,N-N_{B}\}$, the cost is defined as
	\begin{equation}
	J_{N-N_{B}}(x(t+N_{B}),u)=\!\!\!\!\!\!\sum_{k=0}^{N-N_{B}-1}\!\!\!\ell(x_{u}(k),u(k)),
	\end{equation}
\end{subequations}
and the initial state is replaced by a suitable \emph{prediction} of the state at time $t+N_{B}$ using the model equation:
\begin{equation}\label{eq:predicted_state}
x(t+N_{B})=f(x(t+N_{B}-1),\tilde{u}(N-1)).
\end{equation}
\begin{lemma}[Equivalence with reduced-horizon MPC]
	Solving \eqref{eq:blocked_MPC} is equivalent to tackling the reduced horizon problem \eqref{eq:blocked_MPCsimplified}, with initial condition given by the predicted state $x(t+N_{B})$ defined in \eqref{eq:predicted_state}.
\end{lemma}
\begin{proof}
	Since the first $N_{B}$ actions are blocked, the cost in \eqref{eq:blocked_MPCcost} can be rewritten as
	\begin{subequations}
		\begin{equation}
		J_{N}(x(t),u)=J_{N_{B}}(x(t),\tilde{u})+\!\!\sum_{k=N_{B}}^{N-1}\!\!\ell(x_{u}(k),u(k)),
		\end{equation}
		where
		\begin{equation}
		J_{N_{B}}(x(t),\tilde{u})=\sum_{k=0}^{N_{B}-1} \ell(x_{u}(k),\tilde{u}(N\!-\!N_{B}+k)),
		\end{equation}
	\end{subequations}
	and $x_{u}(k)$, for $k=1,\ldots,N_{B}$ can be computed based on the predictive model in \eqref{eq:S}, namely\footnote{Note that both the inputs used for prediction and the associated state satisfy the constraints by design.}
	\begin{align*}
	&x_{u}(1)=x(t+1)\!=\!f(x(t),\tilde{u}(N-N_{B})),\\
	&x_{u}(k)\!=\!f(x_{u}(k-1),\tilde{u}(N\!-\!N_{B}\!+\!k\!-\!1)),~k\!=\!2,\ldots,N_{B},
	\end{align*}
	starting from $x_{u}(0)=x(t)$. Accordingly, $J_{N_{B}}(x(t),\tilde{u})$ is fixed based on this pre-computable state sequence and, thus, it does not impact on the solution of \eqref{eq:blocked_MPC}. At the same time, due to the dynamics of the controlled system, the only state that directly impacts the computation of the optimal input sequence is $x_{u}(N_{B})$. Based on these considerations, the optimization problem in \eqref{eq:blocked_MPC} can be equivalently recast as
	\begin{subequations}\label{eq:blocked_MPCsimpleproof}
		\begin{align}
		&\underset{u}{\mbox{minimize}}~~\sum_{k=N_{B}}^{N-1} \ell(x_{u}(k),u(k))\\
		& \mbox{s.t. }x_{u}(k\!+\!1)\!=\!f(x_{u}(k),u(k)),~ k\!=\!N_{B},\ldots,N\!-\!1,\\
		&\quad~ x_u(N_{B})\!=\!x(t+N_{B}),\\
		&\quad~ x_{u}(k) \in X,~u(k) \in U,~~ k=N_{B},\ldots,N-1.
		\end{align}
	\end{subequations}
	The reduced-horizon MPC problem in \eqref{eq:blocked_MPCsimplified} directly stems from this result based on a simple change of indexes, \emph{i.e.,} considering a new index $k=n-N_{B}$ and shifting the limits of the sum in the loss and the constraints accordingly.
\end{proof}
\begin{algorithm}[!tb]
	\caption{MPC-DMB at time $t$}
	\label{alg:OH-MPC}
	~~\textbf{Input}: initial condition $x(t)$; previous optimal sequence $t-1$, $\{\tilde{u}(k)\}_{k=0}^{N-1}$.
	\vspace*{.1cm}
	\hrule\vspace*{.1cm}
	\begin{enumerate}[label=\arabic*., ref=\theenumi{}]
		\item \textbf{for} $k=0,\ldots,N_{B}-1$
		\begin{enumerate}[label=\arabic*., ref=\theenumi{}]
			\item \textbf{predict} $x(t+k+1)\!=\!f(x(t+k),\tilde{u}(N\!-\!N_{B}+k))$;
		\end{enumerate}
		\item \textbf{solve} the blocked MPC problem in \eqref{eq:blocked_MPCsimplified};
		\item \textbf{extract} the first optimal action $\mu_{N_{B}}(x(t+N_B))\!=\!u(0)$;
		\item \textbf{retrieve} $x(t+N_{B})$;
		\item \textbf{apply} $$u(t+k)=\begin{cases}
		\tilde{u}(N\!-\!N_{B}\!+\!k), &\mbox{if } 0 \leq k \leq N_{B}\!-\!1,\\
		\mu_{N_{B}}(x(t+N_{B})), &\mbox{if } k=N_{B};
		\end{cases}$$
		\item \textbf{update} $\tilde{u}=\{u(t+k)\}_{k=0}^{N_{B}}$;
		\item \textbf{set} $x(0)=x(t+N_{B})$;
		\item \textbf{shift to} $t+N_{B}$;
	\end{enumerate}
	\vspace*{.1cm}
\end{algorithm}
Once the solution of \eqref{eq:blocked_MPC} is computed, the first $N_{B}$ samples of the control sequence (namely, those fixed based on the results of the previous optimization step and the first optimal action over the reduced horizon $N-N_{B}$) are kept, while the rest of the optimal sequence is discarded and the optimization window shifts of $N_{B}$ steps. The resulting MPC-DMB is summarized in Algorithm~\ref{alg:OH-MPC}.

Let us define the optimal value function associated to \eqref{eq:blocked_law} as
\begin{subequations}
\begin{equation}
	V_{N}^{\mathrm{DMB}}(x)=\delta(x)+V_{N-N_{B}}(x_{N_{B}}),
\end{equation}
where
\begin{align}
	& V_{N\!-\!N_{B}}(x_{N_{B}})\!=\!\min_{u \in U}\{V_{N\!-\!N_{B}\!-\!1}(f(x_{N_{B}},u))\!+\!\ell(x_{N_{B}},u)\},\label{eq:reducedHorizon_ovf}\\
	& \delta(x)=\sum_{k=0}^{N_{B}-1}\ell(x_{\tilde{u}}(k),\tilde{u}(N-N_{B}+k)),
\end{align}
with $x_{\tilde{u}}(0)=x$.
\end{subequations}
By comparing the value functions of the receding horizon strategy \eqref{eq:RH_ovf} and the dynamic move blocking approach, we can formalize the following result.
\begin{lemma}[Sub-optimality of MPC-DMB]\label{lemma3}
	Consider a receding horizon $N \in \mathbb{N}$ and a blocking horizon $N_{B} \in \mathbb{N}$, with $1\leq N_{B} \leq N$. Then the following holds:
	\begin{equation}\label{eq:cost_comparison}
		V_{N}(x)\leq V_{N}^{\mathrm{DMB}}(x), \  \forall x \in X.
	\end{equation}
\end{lemma}
\begin{proof}
	According to Bellman's optimality principle, the value function $V_{N}(x)$ can be rewritten as
	\begin{equation}
		V_{N}(x)=\min_{u \in U} \{V_{N-1}(f(x,u))+\ell(x,u)\}.
	\end{equation}
	Therefore, by picking any $\mu(x)$ different from the receding-horizon controller $\mu_{N}(x)$ in \eqref{eq:RH_law}, the following holds by definition:
	\begin{equation}
		V_{N}(x)\!\leq\! \ell(x,\mu(x))+\min_{u \in U}\{V_{N\!-\!2}(f(x_{\mu},u))+\ell(x_{\mu},u)\},
	\end{equation}
	where $x_{\mu}=f(x,\mu(x))$. As a particular case, the previous inequality is thus verified when setting $\mu(x)=\tilde{u}(N-N_{B})$, \emph{i.e.,} considering the first control action resulting from the dynamic move blocking strategy. Analogously, by considering again $\mu(x_{\mu})$ different from $\mu_{N-1}(x_{\mu})$, we further have that
	\begin{align}
		\nonumber V_{N}(x) \leq& \sum_{k=0}^{1}\ell(x_{\mu}(k),\mu(x_{\mu}(k)))+\\
		&~~+\min_{u \in U}\{V_{N\!-\!3}(f(x_{\mu}(2),u))\!+\!\ell(\tilde{x}_{\mu}(2),u)\},
	\end{align}
	where
	\begin{align*}
		&x_{\mu}(0)=x,\\
		&x_{\mu}(k)=f(x_{\mu}(k-1),\mu(x_{\mu}(k-1))),~~k=1,\ldots,2.
 	\end{align*}
 	Once again, specifically $\mu(x_{\mu})=\tilde{u}(N-N_{B}+1)$ is only a particular choice for the sub-optimal law and, thus, the previous inequality holds when using the dynamic move blocking strategy. Exploiting the same reasoning iteratively, it is straightforward to prove that
	\begin{align}
		\nonumber &V_{N}(x)\leq\!\!\!\sum_{k=0}^{N-N_{B}}\ell(x_{\mu}(k),\mu(x_{\mu}(k)))+\\
		&~~+\min_{u \in U}\{V_{N\!-\!N_{B}\!-\!1}(f(x_{\mu}(N_{B}),u))\!+\!\ell(\tilde{x}_{\mu}(N_{B}),u)\},
	\end{align}
	with
	\begin{align*}
		&x_{\mu}(0)=x,\\
		&x_{\mu}(k)=f(x_{\mu}(k-1),\mu(x_{\mu}(k-1))),~~k\!=\!1,\ldots,N-N_{B}.
	\end{align*}
	The proof is thus concluded by simply replacing $\mu(x_{\mu}(k))$ with the $k$-th component $[\mu_{N}^{\mathrm{DMB}}(x)]_{k}$ of $\mu_{N}^{\mathrm{DMB}}(x)$ in \eqref{eq:blocked_law}, for $k=0,\ldots,N_{B}-1$.
\end{proof}
\begin{remark}[The role of $N_B$ for suboptimality]
	Let us consider two possible choices for $N_{B}$, denoted as $N_{B}^{(1)}$ and $N_{B}^{(2)}$ respectively. Let $N_{B}^{(1)}<N_{B}^{(2)}<	N$. As a straightforward consequence of Bellman's optimality principle, it holds that
	\begin{equation}
		V_{N}^{\mathrm{DMB},(1)}(x)\leq V_{N}^{\mathrm{DMB},(2)}(x)
	\end{equation}
	where $V_{N}^{\mathrm{DMB},(i)}(x)$ corresponds to $V_{N}^{\mathrm{DMB}}(x)$ with $N_{B}=N_{B}^{(i)}$, for $i=1,2$, This can be easily proven by following the same reasoning exploited in the proof of Lemma~\ref{lemma3}.
\end{remark}

Let us now introduce the \textquotedblleft ideal\textquotedblright \ value function
\begin{equation}\label{eq:infinite_ovf}
	V_{\infty}(x_{0})=\inf_{u \in \mathcal{U}} J_{\infty}(x(0),u),
\end{equation}
which one would attain by minimizing the (practically unfeasible) infinite-horizon loss in \eqref{eq:infinite_cost}, and let us denote with $V_{\infty}^{\mu_{N}}(x)$ and $V_{\infty}^{\mu_{N}^{\mathrm{DMB}}}(x)$ the optimal value functions obtained when using the receding-horizon controller and the dynamic move blocking strategy over an infinite horizon, respectively, namely,
\begin{subequations}
	\begin{align}
		&V_\infty^{\mu_{N}}(x)=\sum_{k=0}^{\infty} \ell(x_{\mu_{N}}(k),\mu_{N}(x_{\mu_{N}}(k))),\\
		\nonumber & V_\infty^{\mu_{N}^{\mathrm{DMB}}}(x)=\delta(x)+\sum_{k=N_{B}}^{\infty} \ell(x_{\mu_{N}^{\mathrm{DMB}}}(k),\mu_{N}(x_{\mu_{N}^{\mathrm{DMB}}}(k)))=\\
		& ~\qquad\qquad =\delta(x)+V_{\infty}^{\mu_{N_{B}}}(x_{\mu_{N}^{\mathrm{DMB}}}(N_{B})). \label{eq:infty_dmb}
	\end{align}
\end{subequations}
Due to the suboptimality of the receding-horizon controller and the MPC-DMB law, these three functions satisfy the relationship
\begin{equation}\label{eq:infinite_horizons}
	V_{\infty}(x)\leq V_{\infty}^{\mu_{N}}(x)\leq V_{\infty}^{\mu_{N}^\mathrm{DMB}}(x),
\end{equation}
as formalized in the following Lemma.
\begin{lemma}[Ranking of infinite-horizon value functions]\label{lemma:ranking}
	Let $N, N_{B} \in \mathbb{N}$ be the prediction and blocking horizons characterizing the problem in \eqref{eq:blocked_MPC}, with $N_{B}\geq 1$. Then the relationship in \eqref{eq:infinite_horizons} holds for all $x \in X$.
\end{lemma}
\begin{proof}
	Based on the definition of $\mu_{N}^{\mathrm{DMB}}(x)$ and its last component $\mu_{N_{B}}(x)$, it holds that
	\begin{align}
		& V_{N}(x)=V_{N-1}(f(x,\mu_{N}(x)))+\ell(x,\mu_{N}(x)),\\
		\nonumber & V_{N}^{\mathrm{DMB}}(x)\!=\!\delta(x)\!+\!V_{N\!-\!N_{B}\!-\!1}(f(x_{\mu^{\mathrm{DMB}}},\mu_{N_{B}}(x_{\mu^{\mathrm{DMB}}})))+\\
		&\qquad \qquad \qquad \qquad \qquad +\ell(x_{\mu^{\mathrm{DMB}}},\mu_{N_{B}}(x_{\mu^{\mathrm{DMB}}})).
	\end{align}
	Since $N_{B} \neq 0$, $\delta(x) >0$, the MPC-DMB strategy is suboptimal with respect to a standard receding-horizon and the inequality in \eqref{eq:cost_comparison} is satisfied. Based on this result and on the previous definition, \eqref{eq:infinite_horizons} follows when considering $N \rightarrow \infty$.
\end{proof}
\begin{remark}[Corollary of Lemma \ref{lemma:ranking}]
	The suboptimality of the dynamic move blocking strategy is lower-bounded by that of a standard MPC approach, namely
	\begin{equation}
		\frac{V_{\infty}^{\mu_{N}}(x)-V_{\infty}}{V_{\infty}} \leq \frac{V_{\infty}^{\mu_{N}^{\mathrm{DMB}}}(x)-V_{\infty}}{V_{\infty}},\vspace{0.1cm}
	\end{equation}
\end{remark}

Starting from these results, we now follow \cite{Grune2008} toward detecting an upper-bound on the blocked horizon $N_{B}$ preserving the properties characterizing the standard receding horizon strategy and to characterize the suboptimality of the proposed approach. To this end, let us recall the following proposition.
\begin{proposition}[\cite{Grune2008}]
	Consider a generic feedback law $\tilde{\mu: X \rightarrow U}$ and a function $\tilde{V}: X \rightarrow \mathbb{R}_{0}^{+}$ satisfying the inequality
	\begin{equation}\label{eq:key_inequality}
		\tilde{V}(x)\geq \tilde{V}(f(x,\tilde{\mu}(x)))+\alpha \ell(x,\tilde{\mu}(x)),
	\end{equation}
	for some $\alpha \in [0,1]$ and $x \in X$. Then, for all $x \in X$, the following estimate holds:
	\begin{equation}\label{eq:key_ineq2}
		\alpha V_{\infty} \leq \alpha V_{\infty}^{\tilde{\mu}}(x) \leq \tilde{V}(x).
	\end{equation}
\end{proposition}
By focusing on the reduced-horizon problem in \eqref{eq:blocked_MPCsimplified}, whose associated optimal value function is $V_{N-N_{B}}(x_{N_{B}})$ in \eqref{eq:reducedHorizon_ovf}, we introduce the following result.
\begin{lemma}[Infinite vs reduced horizon value functions]\label{eq:key_lemma}
	Consider $N,N_{B} \in \mathbb{N}$, with $N_{B} \in [1,N]$, and the feedback law $\mu_{N_{B}}(x_{N_{B}})$ defined as in \eqref{eq:DMB_lawLast}. Assume that
	\begin{align}\label{eq:condition1}
		\nonumber V_{N\!-\!N_{B}}(x_{N_{B}}^{+})&-V_{N\!-\!N_{B}\!-\!1}(x_{N_{B}}^{+})\leq\\ &\qquad\qquad (1-\alpha)\ell(x_{N_{B}},\mu_{N_{B}}(x_{N_{B}}))),
	\end{align}
	holds for some $\alpha \in [0,1]$ and all $x \in X$, with $x_{N_{B}}^{+}=f(x_{N_{B}},\mu_{N_{B}}(x_{N_{B}})$. Then, $V_{N\!-\!N_{B}}(x_{N_{B}})$ satisfies \eqref{eq:key_inequality} and, thus
	\begin{equation}\label{eq:wanted1}
		\alpha V_{\infty}^{\mu_{N_{B}}}(x_{N_{B}})\leq V_{N\!-\!N_{B}}(x_{N_{B}}),
	\end{equation}
	for all $x \in X$.
\end{lemma}
\begin{proof}
	By exploiting the definition of $V_{N\!-\!N_{B}}(x_{N_{B}})$ in \eqref{eq:reducedHorizon_ovf} and $\mu_{N_{B}}(x_{N_{B}})$ in \eqref{eq:DMB_lawLast} it follows that
	\begin{align*}
		V_{N\!-\!N_{B}}&(x_{N_{B}})=V_{N\!-\!N_{B}\!-\!1}(x_{N_{B}}^{+})+\ell(x_{N_{B}},\mu_{N_{B}}(x_{N_{B}}))\\
		&\geq V_{N\!-\!N_{B}}(x_{N_{B}}^{+})\!-\!(1\!-\!\alpha)\ell(x_{N_{B}},\mu_{N_{B}}(x_{N_{B}}))+\\\
		&\qquad \qquad \qquad \qquad \qquad +\!\ell(x_{N_{B}},\mu_{N_{B}}(x_{N_{B}}))\\
		& = V_{N\!-\!N_{B}}(x_{N_{B}}^{+})+\alpha)\ell(x_{N_{B}},\mu_{N_{B}}(x_{N_{B}})),
	\end{align*}
	where the second inequality holds thanks to \eqref{eq:condition1}. Based on the definition of $x_{N_{B}}^{+}$, from the previous inequality it follows that \eqref{eq:key_inequality} holds, directly yielding \eqref{eq:wanted1}.
\end{proof}
Let us now introduce the following assumptios.
\begin{assumption}[ \cite{Grune2008}]\label{assump:key_assumption}
	For $N \in \mathbb{N}$, there exists $\gamma>0$ such that
	\begin{align}
		&V_{2}(x)\leq (\gamma+1)V_{1}(x),\\
		&V_{k}(x)\leq (\gamma+1)\ell(x,\mu_{k}(x)), k=2,\ldots,N,
	\end{align}
	holds for all $x \in X$.
\end{assumption}
Then, the relationship between $V_{N\!-\!N_{B}}(x_{N_{B}})$ and $V_{N\!-\!N_{B}\!-\!1}(x_{N_{B}})$ can be characterized through $\gamma$ based on the following proposition.
\begin{proposition}[Shrinking of $V_{N\!-\!N_{B}}(x_{N_{B}})$]\label{proposition2}
	Let $N-N_{B}\geq 2$ and Assumption~\ref{assump:key_assumption} be satisfied for the MPC problem \eqref{eq:blocked_MPCsimplified} with horizon $N-N_{B}$. Then, the inequality
	\begin{equation}
	\eta(\gamma,N-N_{B})V_{N\!-\!N_{B}}(x_{N_{B}})\!\leq\!V_{N\!-\!N_{B}\!-\!1}(x_{N_{B}}),
	\end{equation}
	holds for all $x \in X$, where
	\begin{equation}
		\eta(\gamma,N-N_{B})=\left(\!\frac{(\gamma\!+\!1)^{N\!-\!N_{B}\!-\!2}}{(\gamma\!+\!1)^{N\!-\!N_{B}\!-\!2}+\gamma^{N\!-\!N_{B}}}\!\right).\vspace{.2cm}
	\end{equation}
\end{proposition}
	The proof of this proposition straightforwardly follows from that of \cite[Proposition 4.4]{Grune2008}, by replacing the full horizon $N$ with the difference $N-N_{B}$ and it is thus omitted.
\begin{remark}[Constraining $N_{B}$]
Through Proposition~\ref{proposition2}, we obtain a first upper bound on $N_{B}$. Indeed, the blocking horizon $N_{B}$ must satisfy
\begin{equation}\label{eq:first_bound}
	N_{OP} \leq N_{B} \leq N-2,
\end{equation}
for the result in Proposition~\ref{proposition2} to hold. In turn, this enforces $N$ to be at least equal to $N_{OP}+2$.
\end{remark}
\begin{remark}[A note about Assumption~\ref{assump:key_assumption}]
	The satisfaction of Assumption~\ref{assump:key_assumption} can be checked by analyzing the running cost $\ell(\cdot)$. In particular, as detailed in \cite[Proposition 4.7]{Grune2008}, the loss should satisfy
	\begin{equation*}
		\ell(x,u)\geq \alpha W(x),
	\end{equation*}
	with $\alpha >0$ and $W: X \rightarrow \mathbb{R}_{0}^{+}$, while also verifying an exponential controllability condition.
\end{remark}
By leveraging this result, we can now characterize the suboptimality of $\mu_{N_{B}}(x_{N_{B}})$ as follows.
\begin{theorem}[Asymptotic stability]\label{thm1}
	Let $N,N_{B} \in \mathbb{N}$ be the prediction and blocking horizon of the MPC-DMB problem in \eqref{eq:blocked_MPC}, respectively, and let $N_{B}$ satisfy \eqref{eq:first_bound}. Let $\gamma>0$, assume that
	\begin{equation}\label{eq:assumption}
	(\gamma+1)^{N-N_{B}-2}>\gamma^{N-N_{B}},
	\end{equation}
	and that Assumption~\ref{assump:key_assumption} holds for these $\gamma$ and for the reduced horizon $N-N_{B}$. Then, for all $x \in X$, it holds that
	\begin{equation}\label{eq:performance_bound1}
		V_{\infty}^{\mu_{N_{B}}}(x_{N_{B}}) \leq \frac{(\gamma+1)^{N\!-\!N_{B}\!-\!2}}{(\gamma+1)^{N\!-\!N_{B}\!-\!2}\!-\!\gamma^{N-N_{B}}} V_{\infty}(x_{N_{B}}),
	\end{equation}
	and, as a consequence,
	\begin{equation}\label{eq:performance_bound2}
		\frac{V_{\infty}^{\mu_{N_{B}}}(x_{N_{B}})\!-\!V_{\infty}(x_{N_{B}})}{ V_{\infty}(x_{N_{B}})} \!\leq\! \frac{\gamma^{N\!-\!N_{B}}}{(\gamma+1)^{N\!-\!N_{B}\!-\!2}\!-\!\gamma^{N\!-\!N_{B}}}.\vspace{.1cm}
	\end{equation}
\end{theorem}
\begin{proof}
	The proof follows the same steps of \cite[Proof of Theorem 4.5]{Grune2008}. In particular, from proposition~\ref{proposition2}, it follows that
	\begin{align*}
		\Delta V_{N\!-\!N_{B}}(x_{N_{B}})&\leq (\eta^{\!-1}(\gamma,n-n_{b})-1)V_{N\!-\!N_{B}\!-\!1}(x_{N_{B}})\\
		& = \frac{\gamma^{N-N_{B}}}{(\gamma+1)^{N-N_{B}-2}}V_{N\!-\!N_{B}\!-\!1}(x_{N_{B}}).
	\end{align*}
	where $\Delta V_{N\!-\!N_{B}}(x_{N_{B}})=V_{N\!-\!N_{B}}(x_{N_{B}})-V_{N\!-\!N_{B}\!-\!1}(x_{N_{B}})$. Since Assumption~\ref{assump:key_assumption} implies that\footnote{The reader is referred to \cite{Grune2008} for a detail proof of this inequality.}
	\begin{align*}
		&\Delta V_{N\!-\!N_{B}}(f(x_{N_{B}},\mu_{N_{B}}(x_{N_{B}})) \leq \\
		&\qquad \qquad \qquad \qquad\frac{\gamma^{N-N_{B}}}{(\gamma\!+\!1)^{N\!-\!N_{B}\!-\!2}} \ell(x_{N_{B}},\mu_{N_{B}}(x_{N_{B}})),
	\end{align*}
	in turn, yielding the inequality in \eqref{eq:condition1}. Therefore, Lemma~\ref{lemma3} can be applied by setting
	\begin{equation}
	\alpha=1-\frac{\gamma^{N\!-\!N_{B}}}{(\gamma+1)^{N\!-\!N_{B}\!-\!2}}=\frac{(\gamma+1)^{N\!-\!N_{B}\!-\!2}-\gamma^{N\!-\!N_{B}}}{(\gamma+1)^{N\!-\!N_{B}\!-\!2}}.
	\end{equation}
	Hence, \eqref{eq:performance_bound1} results directly from \eqref{eq:key_ineq2} that, in turn, yields the bound on the relative difference between infinite value functions in \eqref{eq:performance_bound2}, thus concluding the proof.
\end{proof}
Thanks to our assumptions about the features of the running cost and the state space $X$, Theorem~\ref{thm1} implies \textit{asymptotic stability} of the compact set $A$ with respect to which $\ell(\cdot)$ is positive definite and proper, starting from the initial state $x_{N_{B}}$ and solving the reduced-horizon problem in \eqref{eq:blocked_MPCsimplified} (as discussed in \cite{Grune2008}). As a direct consequence of this Theorem~\ref{thm1}, we can further formalize our bound \eqref{eq:first_bound} on the blocking horizon as follows.
\begin{lemma}[Further constraining $N_B$]
	Given $N,N_{B} \in \mathbb{N}$, let Assumption~\ref{assump:key_assumption} be satisfied by $V_{N\!-\!N_{B}}(x_{N_{B}})$. Then, for the performance bounds \eqref{eq:performance_bound1}-\eqref{eq:performance_bound2} to hold, the blocking horizon should be upper-bounded as follows:
	\begin{equation}\label{eq:upper_bound}
		N_{B} < N-2\frac{\log(\gamma+1)}{\log(\gamma+1)-\log(\gamma)}. \vspace{.2cm}
	\end{equation}
\end{lemma}
\begin{proof}
	The proof is easily obtained by simply manipulating the assumption on $\delta$ and the reduced horizon $N-N_{B}$ required by Theorem~\ref{thm1} (see \eqref{eq:assumption}).
\end{proof}

By relying on these results for the reduced horizon problem \eqref{eq:blocked_MPCsimplified}, we can now state the following sub-optimality bound for the MPC-DMB scheme.
\begin{theorem}[Suboptimality bound for MPC-DMB]\label{thm2}
	Let $N,N_{B} \in \mathbb{N}$ be the prediction and blocking horizon of the MPC-DMB problem in \eqref{eq:blocked_MPC}, respectively, and let $N_{B}$ satisfy \eqref{eq:upper_bound}. Suppose that there exists a $\gamma>0$ such that \eqref{eq:assumption} and Assumption~\ref{assump:key_assumption} hold, with the latter being verified for the reduced horizon $N-N_{B}$. Then, for all $x \in X$, it holds that
	\begin{equation}\label{eq:performance_boundOverall}
		\frac{V_{\infty}^{\mu_{N}^{\mathrm{DMB}}}\!\!\!\!(x)\!-\!V_{\infty}(x)}{ V_{\infty}(x)} \!\leq\! \frac{\gamma^{N\!-\!N_{B}}}{(\gamma+1)^{N\!-\!N_{B}\!-\!2}\!-\!\gamma^{N\!-\!N_{B}}}\!+\!\!\frac{\delta(x)}{V_{\infty}(x)}.\vspace{.1cm}
	\end{equation}
\end{theorem}
\begin{proof}
	Thanks to our assumptions, the bound in \eqref{eq:performance_bound1} holds. Therefore,
	\begin{equation*}
		V_{\infty}^{\mu_{N_{B}}}(x_{N_{B}})+\delta(x)-\delta(x) \leq  \upsilon(\gamma,N-N_{B}) V_{\infty}(x_{N_{B}}),
	\end{equation*}
	with $\upsilon(\gamma,N-N_{B})=\frac{(\gamma+1)^{N\!-\!N_{B}\!-\!2}}{(\gamma+1)^{N\!-\!N_{B}\!-\!2}\!-\!\gamma^{N-N_{B}}}$ which further implies that
	\begin{equation}\label{eq:bound_1}
			V_{\infty}^{\mu_{N}^{\mathrm{DMB}}}(x)-\delta(x) \leq  \upsilon(\gamma,N-N_{B}) V_{\infty}(x_{N_{B}}),
	\end{equation}
	according to \eqref{eq:infty_dmb}. Since $\ell(x,u) \in \mathbb{R}^{+}_{0}$ for all $(x,u) \in X \times U$ and $x_{N_{B}}$ is obtained as the $N_{B}$ steps ahead prediction of the state of \eqref{eq:S} starting from $x$, the following further holds
	\begin{align*}
		V_{\infty}(x_{N_{B}})&=\inf_{u \in \mathcal{U}} \sum_{k=N_{B}}^{\infty}\ell(x_{u}(k),u(k))\\
		& \leq \inf_{u \in \mathcal{U}} \sum_{k=0}^{\infty}\ell(x_{u}(k),u(k))=V_{\infty}(x),
	\end{align*}
	where $x_{u}(N_{B})=X_{N_{B}}$ while $x_{u}(0)=x$. Hence, from \eqref{eq:bound_1} we obtain that
		\begin{equation}\label{eq:bound_2}
		V_{\infty}^{\mu_{N}^{\mathrm{DMB}}}(x)-\delta(x) \leq  \upsilon(\gamma,N-N_{B}) V_{\infty}(x).
	\end{equation}
	Straightforward manipulations of this inequality, leads to
	\begin{equation}
		V_{\infty}^{\mu_{N}^{\mathrm{DMB}}}(x)-V_{\infty}(x) \leq  \upsilon(\gamma,N-N_{B})+\delta(x),
	\end{equation}
	that, by dividing both terms for $V_{\infty}(x)$ leads to \eqref{eq:performance_boundOverall} and, thus, concludes the proof.
\end{proof}
This Theorem~\ref{thm2} formalizes a rather intuitive result. Indeed, \eqref{eq:performance_boundOverall} implies that the level of suboptimality of the controller resulting from the MPC-DMB scheme is linked to both the properties of control action at time $t+N_{b}$ (that is optimized) and the value of the cumulative loss associated with those steps where the input is blocked.

%% file: sections/04-App-example.tex
\section{Numerical example}
\label{sec:POC-example}

The scope of this brief numerical section is to show that, at least in a simple simulation case study, the DMB approach is feasible way to implement MPC under computational constraints, without leading to detrimental performance. Let the system to be controlled be the linear time-invariant plant
\begin{equation}
\left\{
\begin{array}{rcl}
x_1(t+1) &=& 0.9 \cdot x_1(t) + 0.1 \cdot u(t)\\
x_2(t+1) &=& 0.6 \cdot x_1(t) + 0.4 \cdot x_2(t)\\
y(t)     &=& x_2(t),
\end{array}
\right.
\end{equation}
while the control problem to solve is as indicated in \eqref{eq:infinite_cost}, where the quadratic cost
\[\ell(x_u(t),u(t))=x(t)^TQx(t)+u(t)^TRu(t),\]
is employed, $R = 1$, and 
$$Q = \begin{bmatrix}10 & 0\\0 & 100\end{bmatrix}.$$
Let the control variable $u$ be constrained so that it must lie within the [-1,1] range. No constraints on the state $x$ are instead imposed, namely $X \equiv \mathbb{R}^2$. We consider the achieved closed-loop properties in terms of \textit{regulation to zero}, starting from an initial condition with $x_1=1$ and $x_2=-1$. The prediction horizon $N$ is set equal to $6$ steps.

Let us assume that the computation time needed for the solution of the control problem amounts to three steps, thus $N_B=N_{OP}=3$. For a fair comparison, we assess DMB MPC as compared to a traditional MPC control with reduced horizon $N-N_B=3$.

The time histories of the state trajectories are illustrated in Figure~\ref{fig:ex1}, where it can be clearly seen that the error between the (unfeasible) receding-horizon solution and that of the (feasible) DMB-MPC approach remains limited, even if MPC is obviously faster at the end of the transient, due to a more rapid update of the control action. This is also confirmed by the small difference between the Root Mean Square Errors (RMSE) reported in Table \ref{tab:rmse} for the two strategies.

\begin{figure}
	\includegraphics[width=1\columnwidth]{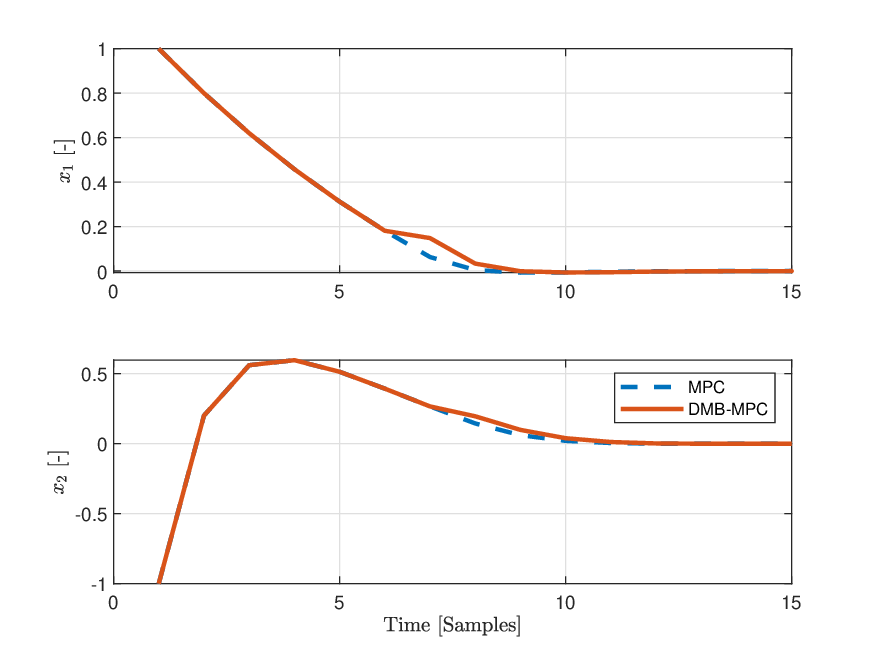}
	\caption{State trajectories in the example of Section \ref{sec:POC-example} for a traditional MPC strategy and the proposed DMB rationale.}
	\label{fig:ex1}
\end{figure}

\begin{table}
\begin{center}
	\begin{tabular}{|c | c|}
\hline
		\textbf{MPC} & 15.9134  \\
		\hline 
	\textbf{DMB MPC} & 16.1334\\
		\hline 
	\end{tabular}
\end{center}\caption{RMSE for MPC and DMB MPC.}\label{tab:rmse}
\end{table}

%% file: sections/05-Conclusions.tex
\section{Conclusions and future work}
\label{sec:Conclusions}

In this paper, we introduced a \textit{dynamic move blocking} MPC approach to comply with the case in which the optimisation cannot be solved in one sampling period, without changing the formulation of the problem at hand. The issue occurs whenever computational resource limitations are relevant, but the sampling frequency cannot be decreased due to physical constraints.

The proposed strategy is simple yet effective, as it exploits the optimal solution already available (the one computed at the previous round) while the new optimization is ongoing. The resulting problem is shown to be equivalent to a reduced horizon MPC. Asymptotic stability as well as error bounds have been proven under some mild assumptions.

Future work will be directed toward an extensive experimental assessment of the proposed strategy, as well as a comparison with the other empirical alternatives to deal with the problem of limited resources.

%% file: main.bbl
\begin{thebibliography}{10}

\bibitem{alessio2009survey}
Alessandro Alessio and Alberto Bemporad.
\newblock A survey on explicit model predictive control.
\newblock {\em Nonlinear Model Predictive Control: Towards New Challenging
  Applications}, pages 345--369, 2009.

\bibitem{allgower2012nonlinear}
F.~Allg{\"o}wer and A.~Zheng.
\newblock {\em Nonlinear model predictive control}.
\newblock Birkh{\"a}user, 2012.

\bibitem{berberich2022linear}
Julian Berberich, Johannes K{\"o}hler, Matthias~A M{\"u}ller, and Frank
  Allg{\"o}wer.
\newblock Linear tracking {MPC} for nonlinear systems—part {I}: The
  model-based case.
\newblock {\em IEEE Transactions on Automatic Control}, 67(9):4390--4405, 2022.

\bibitem{borrelli2017predictive}
Francesco Borrelli, Alberto Bemporad, and Manfred Morari.
\newblock {\em Predictive control for linear and hybrid systems}.
\newblock Cambridge University Press, 2017.

\bibitem{cagienard2007move}
Raphael Cagienard, Pascal Grieder, Eric~C Kerrigan, and Manfred Morari.
\newblock Move blocking strategies in receding horizon control.
\newblock {\em Journal of Process Control}, 17(6):563--570, 2007.

\bibitem{gondhalekar2007recursive}
Ravi Gondhalekar and Jun-ichi Imura.
\newblock Recursive feasibility guarantees in move-blocking mpc.
\newblock In {\em 2007 46th IEEE Conference on Decision and Control}, pages
  1374--1379. IEEE, 2007.

\bibitem{Grune2008}
L.~Grune and A.~Rantzer.
\newblock On the infinite horizon performance of receding horizon controllers.
\newblock {\em IEEE Transactions on Automatic Control}, 53(9):2100--2111, 2008.

\bibitem{leva2022overlapping}
Alberto Leva, Simone Formentin, and Silvano Seva.
\newblock Overlapping-horizon {MPC}: A novel approach to computational
  constraints in real-time predictive control.
\newblock In {\em Third Workshop on Next Generation Real-Time Embedded Systems
  (NG-RES 2022)}. Schloss Dagstuhl-Leibniz-Zentrum f{\"u}r Informatik, 2022.

\bibitem{schwickart2016flexible}
Tim Schwickart, Holger Voos, Mohamed Darouach, and Souad Bezzaoucha.
\newblock A flexible move blocking strategy to speed up model-predictive
  control while retaining a high tracking performance.
\newblock In {\em 2016 European Control Conference (ECC)}, pages 764--769.
  IEEE, 2016.

\bibitem{shekhar2015optimal}
Rohan~C Shekhar and Chris Manzie.
\newblock Optimal move blocking strategies for model predictive control.
\newblock {\em Automatica}, 61:27--34, 2015.

\end{thebibliography}
